\documentclass[12pt,reqno]{amsart}

\usepackage{amsmath}
\usepackage{amsthm}
\usepackage{amsfonts}
\usepackage{amssymb}
\usepackage{latexsym}
\usepackage{amscd}
\usepackage{stmaryrd}
\usepackage{amsbsy}

\allowdisplaybreaks
\newcommand{\be}{\begin{eqnarray*}}
\newcommand{\en}{\end{eqnarray*}}
\newcommand{\bea}{\begin{eqnarray}}
\newcommand{\ena}{\end{eqnarray}}

\newcommand{\C}{{\mathbb C}}
\newcommand{\Z}{{\mathbb Z}}

\newcommand{\slt}{\mathfrak{sl}_2}

\newcommand{\ka}{\kappa}

\def\g{\operatorname{\mathfrak{g}}}

\numberwithin{equation}{section}

\theoremstyle{plain}
\newtheorem{thm}{Theorem}[section]
\newtheorem{cor}[thm]{Corollary}
\newtheorem{lem}[thm]{Lemma}
\newtheorem{prop}[thm]{Proposition}
\newtheorem{dfn}[thm]{Definition}
\newtheorem{exmp}[thm]{Example}
\newtheorem{re}[thm]{Remark}

\begin{document}

\begin{flushleft}
Copyright (2011) American Institute of Physics. This article may be downloaded for personal use only. Any other use requires prior permission of the author and the American Institute of Physics.

The following article appeared in Journal of Mathematical Physics (Vol.52, Issue 8) and may be found at 
 http://link.aip.org/link/?JMP/52/083509
 
\end{flushleft}

\bigskip

\title{Hypergeometric solutions to Schr\"odinger equations for the quantum Painlev\'e equations}
\date{}
\author{H.~Nagoya}
\address{HN:Department of Mathematics, Kobe University,
 Kobe 657-8501, 
 Japan \\
 Research Fellow of the Japan Society for the Promotion of Science}\email{nagoya@math.kobe-u.ac.jp}

\maketitle

\begin{abstract}
We consider Schr\"odinger equations for the quantum Painlev\'e equations. 
We present hypergeometric solutions of the Schr\"odinger equations for the quantum Painlev\'e equations, as particular solutions.  
We also give a representation theoretic correspondence between Hamiltonians of the Schr\"odinger equations for 
the quantum Painlev\'e equations and those of the KZ equation 
or the confluent KZ equations. 

\end{abstract}

\section{Introduction}
Quantizations of the Painlev\' e equations with  affine Weyl group symmetries, in the Heisenberg picture,  
were proposed and studied 
in  \cite{JNS}, \cite{N1}, 
 \cite{N4}, and \cite{RNGT}.  
We have called these quantizations, quantum Painlev\'e equations. 
Although it is well-known that the Painlev\'e equations admit exact solutions, such as hypergeometric solutions, algebraic or 
rational solutions,  
no exact solutions of the quantum Painlev\'e equations were given in the previous studies.

Aiming to construct exact solutions to the quantum Painlev\'e equations, together with M. Jimbo and J. Sun, 
we introduced confluent KZ equations for $\slt$, 
irregular singular versions of the KZ equations \cite{JNS}. 
The confluent KZ equations for $\slt$ have  integral formulas for solutions, which take values in confluent Verma 
modules. We derived the quantum Painlev\'e equations in a formal algebraic way  from the Heisenberg version of the 
confluent KZ equations.  In the derivation process, we need to invert some elements. For the 
confluent KZ equations associated with 
highest weight-type modules, this invertibility fails. Hence the integral solutions 
to the confluent KZ equations do not give rise to solutions to the quantum Painlev\'e equations. 

However, 
the work \cite{JNS} gives us some insight into a possible link between Hamiltonians of the (confluent)
 KZ equations and those of Schr\"odinger equations for the quantum Painlev\'e equations. Because  
 the elements that we have to invert,  do not come in Hamiltonians derived from the Hamiltonians of the Heisenberg version 
 of the confluent KZ equations, in the reduction process. Unfortunately, the result of \cite{JNS} can not 
 explain immediately a relation between the (confluent) KZ equations and Schr\"odinger equations of the 
 quantum Painlev\'e equations. We study them in the present paper.

We consider  following 
 Schr\"odinger equations for the quantum Painlev\'e equations:
\begin{equation}\label{eq sch}
\hbar\frac{\partial }{\partial t}\Phi(x,t)=\widehat{H}_\mathrm{J}\left(
x,\hbar\frac{\partial}{\partial x}, t
\right)
\Phi(x,t) \quad (\mathrm{J=II, III, IV, V, VI}), 
\end{equation}
where $\hbar$ is a parameter in $\C$ and Hamiltonians $\widehat{H}_\mathrm{J}$ are obtained from polynomial 
Hamiltonians of 
the Painlev\'e equations by substituting the operators $x$, $\hbar \partial / \partial x$ into the canonical coordinates 
$q$, $p$ (see Definition \ref{def hamiltonian}). 
The Schr\"odinger equations for the quantum Painlev\'e equations of type $\mathrm{III, V, VI}$ 
  have appeared as 
differential equations satisfied by correlation functions of  two dimensional 
conformal field theory and expected to be differential equations satisfied by 
 instanton partition functions 
of $SU(2)$ gauge theories \cite{AFKMY}.  
 
 As an example, the quantized Hamiltonian $\widehat{H}_\mathrm{II}$ is  
 \begin{equation*}
\widehat{H}_\mathrm{II}=\frac{1}{2}\left(\hbar
\frac{\partial}{\partial x}
\right)^2-\left(
x^2+\frac{t}{2}
\right)\hbar\frac{\partial}{\partial x}+ax, 
\end{equation*}
where $a$ is a complex parameter. It is easy to see that the different ordering of the 
operators $x$ and $\partial/ \partial x$ can be absorbed in the parameter $a$. 
The other quantized Hamiltonians $\widehat{H}_\mathrm{J}$ also have same property. 
 
 Another property of the quantized Hamiltonians is that they act on $\bigoplus_{i=0}^m\C x^i$ if $a=m\hbar$ 
 with a nonnegative integer $m$.  Consequently, we can consider polynomial solutions 
 \begin{equation*}
\Phi(x,t)=\sum_{i=0}^m\varphi_i(t)x^i
\end{equation*}
and the Schr\"odinger equations for the quantum Painlev\'e equations \eqref{eq sch} become linear differential systems for $\varphi_i(t)$ ($i=0,1,\ldots,m$). We remark that this property is also satisfied by Schr\"odinger equations for a quantization of 
other Painlev\'e systems or isomonodromy deformations, for example, the Garnier system 
\cite{Garnier}. 

When $m=1$ and $\mathrm{J=II}$, the linear differential system is
\begin{equation*}
\frac{d}{dt}\varphi_0(t)=- \frac{t}{2}\varphi_1(t),
\quad
\frac{d}{dt}\varphi_1(t)= \varphi_0(t).  
\end{equation*}
 Namely, $\varphi_1(t)$ satisfies the Airy equation. For the cases of $\mathrm{J=III, IV, V, VI}$, 
 $\varphi_1(t)$ satisfy the Bessel equation, the Hermite-Weber equation, the Kummer 
  equation and the Gauss hypergeometric equation, respectively. 
 
 In terms of the integral representation,  a polynomial solution $\Phi(x,t)$ for $m=1$ and $\mathrm{J=II}$  
 can be expressed as 
 \begin{equation*}
\Phi(x,t)=\int_\Gamma \exp\left(-ut-\frac{2}{3}u^3\right)(x-u)du
\end{equation*}
for an appropriate cycle $\Gamma$. 
For general $m\in\Z_{\ge 0}$, polynomial solutions are expressed by 
a generalization of the integral formula above. 

 Let $\rho_\mathrm{J}(u)$ ($\mathrm{J=II, III, IV, V, VI}$) be the master functions of the integral formulas of the Airy function, 
the Bessel function, the Hermite-Weber function, the Kummer 
  function and the Gauss hypergeometric function, respectively (see Definition 
 \ref{def master function}).

\begin{thm}
For $m\in\Z_{\ge 0}$, the integral formula
\begin{equation}
\Phi^\mathrm{J}_m(x,t)=\int_\Gamma \prod_{1\le i<j\le m}(u_i-u_j)^{2\hbar}\prod_{i=1}^m \rho_\mathrm{J}(u_i)
(x-u_i)du_i
\end{equation}
is a solution of the Schr\"odinger equations for the quantum Painlev\'e equation of type $\mathrm{J}$ with  $a=m\hbar$. 
Here, $\Gamma$ is an $m$-cycle of the homology group 
determined by the integrand.  
\end{thm}

These integral formulas have been studied in various fields, such as, 
two dimensional conformal field theory,  random matrix theory(in particular  the $\beta$-ensembles, for example, see \cite{BEMP} and 
references therein), the theory of 
the orthogonal polynomials, and the theory of hypergeometric functions. 
In particular, if $\hbar=1$ and $\Gamma=\prod_{i=1}^m\Gamma_1$, where 
$\Gamma_1$ is an appropriate cycle for $\Phi^\mathrm{J}_1(x,t)$,  the integral formulas 
$\Phi^\mathrm{J}_m(x,t)$ ($m\in \Z_{\ge 0}$) are orthogonal polynomials in $x$ with respect to 
the weight function $\rho_\mathrm{J}(u)$. Because the orthogonal polynomials admit determinant 
representations (for example, see \cite{Szego}, Chapter II), we have 
\begin{cor}
If $\hbar=1$, then  determinant formulas
\begin{equation}\label{eq determinant formula}
P^\mathrm{J}_m(x,t)=\det
\begin{pmatrix}
\left(
\tau_\mathrm{J}^{(i+j-2)}
\right)_{1\le j\le m+1}^{1\le i\le m}
\\
\left(
x^{j-1}
\right)_{1\le j\le m+1}
\end{pmatrix}=
\left|
 \begin{matrix}
 \tau_\mathrm{J} & \tau_\mathrm{J}^{(1)} & \cdots & \tau_\mathrm{J}^{(m)}
  \\
  \tau_\mathrm{J}^{(1)} & \tau_\mathrm{J}^{(2)} & \cdots & \tau_\mathrm{J}^{(m+1)}
  \\
  \vdots &  &  & \vdots
  \\
  \tau_\mathrm{J}^{(m-1)} & \cdots & \cdots & \tau_\mathrm{J}^{(2m-1)}
  \\
  1 &x& \cdots & x^m
  \end{matrix}\right|,
\end{equation}
where 
$\tau_\mathrm{J}=\tau_\mathrm{J}(t)=\int_{\Gamma_1} \rho(u)du$,  $\tau_\mathrm{J}^{(i)}=\tau_\mathrm{J}^{(i)}(t)=\int_{\Gamma_1} u^i\rho(u)du $, 
are solutions of the Schr\"odinger equations for the quantum Painlev\'e equations of type  $\mathrm{J=II, III, IV, V, VI}$ with $a=m$. 
\end{cor}

We note that the coefficients of $x^m$ in $P^\mathrm{J}_m(x,t)$ 
\begin{equation*}
(\tau_\mathrm{J})_m=\left|
 \begin{matrix}
 \tau_\mathrm{J} & \tau_\mathrm{J}^{(1)} & \cdots & \tau_\mathrm{J}^{(m-1)}
  \\
  \tau_\mathrm{J}^{(1)} & \tau_\mathrm{J}^{(2)} & \cdots & \tau_\mathrm{J}^{(m)}
  \\
  \vdots &  &  & \vdots
  \\
  \tau_\mathrm{J}^{(m-1)} & \cdots & \cdots & \tau_\mathrm{J}^{(2m-2)}
  \end{matrix}\right|
\end{equation*}
are tau functions of 
the Painlev\'e equations $\mathrm{P_J}$ $(\mathrm{J=II, III, IV, V, VI})$ (for example, see \cite{FW}, \cite{Masuda} and references therein). 
This is because if $\hbar=1$, then the Schr\"odinger equations for the quantum Painlev\'e equations are related to isomonodromy deformations for 
the Painlev\'e equations  \cite{Suleimanov 1}.  

As mentioned above, the quantum Painlev\'e equations are related to 
 two dimensional conformal field theory. It is known that 
the Knizhnik-Zamolodchikov (KZ) equation can be viewed as a quantization of the Schlesinger equations, 
which is isomonodromy deformation with regular singularities 
\cite{H}, \cite{R}. 
On the other hand, the sixth Painlev\'e equation is derived from the 
Schlesinger equation and the other Painlev\'e equations are derived from 
the irregular Schlesinger equations \cite{JMU}. 
We give a representation theoretic correspondence between the 
(confluent) KZ equations and the Schr\"odinger equations for the quantum Painlev\'e equations. 

Let $M_1$, $M_2$, and $M_3$ be Verma modules with highest weights $\gamma^{(1)}_0$, $\gamma^{(2)}_0$, and 
$\gamma^{(3)}_0$ for $\frak{sl}_2$.  
We consider the KZ equation on 
the tensor product of three Verma modules $M=M_1\otimes M_2 \otimes M_3$ 
with three points $0$, $t$, $1$. Let  
 $H_\mathrm{KZ}$ the Hamiltonian of the KZ equation. 
 \begin{thm}\label{thm intro kz}
 For $\gamma^{(i)}_0\not\in\Z$ ($i=1,2,3$) and $m\in\mathbb{Z}_{\ge 0}$,  
the action of $H_\mathrm{KZ}$ 
on  the space of singular vectors of weight $\sum_{i=1}^3 \gamma^{(i)}_0-2m$ is 
equivalent to the action of the quantized Hamiltonian $\widehat{H}_\mathrm{{VI}}$ on the subspace 
$\oplus_{i=0}^m\mathbb{C}x^i$  of the polynomial ring $\mathbb{C}[x]$ with $a=m\hbar$.   
\end{thm}

In  the cases of $\mathrm{J}=\mathrm{II,III,IV,V}$, we consider the confluent KZ 
equations  
defined in \cite{JNS}.  
\begin{thm}\label{thm intro ckz}
For $m\in\mathbb{Z}_{\ge 0}$,  
the actions of Hamiltonians of the certain confluent KZ equations are equivalent to the actions 
of the quantized Hamiltonians $\widehat{H}_\mathrm{J}$ ($\mathrm{J}=\mathrm{II,III,IV,V}$) with 
$a=m\hbar$. 
\end{thm}
In these cases, we do not consider the space of singular vectors.

The remainder of this paper is organized as follows. In section 2, we present integral formulas 
and prove that the integral formulas are solutions to the Schr\"odinger equations for the quantum Painlev\'e equations. 
Moreover, we give determinant formulas for solutions. 
In section 3, we recall the definition of the 
(confluent) KZ equations corresponding to the Schr\"odinger equations for the quantum Painlev\'e equations 
and give a representation theoretic correspondence 
between the (confluent) KZ equations and the Schr\"odinger equations for the quantum Painlev\'e equations. 
We also give a symmetry of the Schr\"odinger equations for the quantum Painlev\'e equations with respect to $\hbar \to -\hbar$ in 
quantized Hamiltonians.

\section{Integral formula}
In this section, we present integral formulas taking values in $\bigoplus_{i=0}^m \mathbb{C}x^i$ ($m\in\Z_{\ge 0} $) and 
show that they are solutions to the quantum Painlev\'e equations.

\begin{dfn}\label{def hamiltonian}
Quantized Hamiltonians 
$\widehat{H}_\mathrm{J}$ are defined as
\begin{align*}
&t(t-1)\widehat{H}_{\mathrm{VI}}\left(
x, \hbar \frac{\partial}{\partial x}, a, b, c, d,  t
\right)=x(x-1)(x-t)\left(\hbar\frac{\partial}{\partial x}\right)^2
\\&
-\left(
(a+b)(x-1)(x-t)+cx(x-t)
+dx(x-1)
\right)\hbar\frac{\partial}{\partial x}
+(b+c+d+\hbar)a (x-t) , 
\\
&t\widehat{H}_\mathrm{V}\left(
x, \hbar \frac{\partial}{\partial x}, a, b, c,  t
\right)=x(x-1)\left(\hbar \frac{\partial}{\partial x}\right)^2
+\left(tx^2 -(b+c+t)x+b\right)\hbar\frac{\partial}{\partial x}
\\
&+
a(b+c-a+\hbar +t-tx),
\\
&\widehat{H}_\mathrm{IV}\left(
x, \hbar \frac{\partial}{\partial x}, a, b, t
\right)=x\left(\hbar\frac{\partial}{\partial x}\right)^2+(x^2 -tx -b)
\hbar\frac{\partial}{\partial x} -
 ax-(a-2b)t, 
\\
&t\widehat{H}_\mathrm{III}\left(
x, \hbar \frac{\partial}{\partial x}, a, b, t
\right)= x^2\left(\hbar\frac{\partial}{\partial x}\right)^2 -\left(x^2+bx +t\right)\hbar\frac{\partial}{\partial x}
+ ax,
\\
&\widehat{H}_\mathrm{II}\left(
x, \hbar \frac{\partial}{\partial x}, a, t
\right)=\frac{1}{2}\left(\hbar\frac{\partial}{\partial x}\right)^2
-\left(x^2+\frac{t}{2}\right)\hbar
\frac{\partial}{\partial x}+ ax, 
\end{align*}
where $a,b,c,d\in\mathbb{C}$. 
\end{dfn}
\begin{re}
The spectral problem of 
the Hamiltonian $\widehat{H}_{\mathrm{VI}}$ 
is the Heun equation, which is equivalent to the $BC_1$ Inozemtsev model. 
This was known as the
Painlev\'e-Calogero correspondence \cite{LO}, \cite{T}.

\end{re}

As mentioned in Introduction, the Schr\"odinger equations for the quantum Painlev\'e equations have polynomial solutions in $x$ 
because of the following Proposition. 
\begin{prop}\label{prop a}
The quantized Hamiltonians $\widehat{H}_\mathrm{J}$ act on 
$\bigoplus_{i=0}^m \mathbb{C}x^i$ ($m\in\Z_{\ge 0} $) if and only if 
$$
\left\{
\begin{tabular}{cc}
$a=m\hbar$ \ or \  $ b+c+d=(m-1)\hbar$ & $\mathrm{J=VI}$, \\ $ a=m\hbar$ & $\mathrm{J=II,III,IV,V}$. 
\end{tabular}\right.
$$
\end{prop}
\begin{proof}
It is sufficient to compute the action of $\widehat{H}_J$ on $x^m$. If  $\mathrm{J=II}$, then, 
\begin{equation*}
\widehat{H}_\mathrm{II}x^m=\left(
a-m\hbar
\right)x^{m+1}-\frac{t}{2}m\hbar x^m+\frac{1}{2}m(m-1)\hbar^2x^{m-2}. 
\end{equation*}
Hence, $\widehat{H}_\mathrm{II}$ acts on $\bigoplus_{i=0}^m \mathbb{C}x^i$ if and only if $a=m\hbar$. 

For the other cases, we can prove Proposition \ref{prop a} in the same way. 
\end{proof}

From Proposition \ref{prop a}, the Schr\"odinger equations for  the quantum Painlev\'e equations have polynomial solutions 
$\Phi^\mathrm{J}_m(x,t)=\sum_{i=0}^m\varphi_i(t)x^i$ and then become linear differential systems of $\varphi_i(t)$. We show that 
  polynomial solutions $\Phi^\mathrm{J}_m(x,t)$ are expressed in terms of  integral formulas of hypergeometric type. 

\begin{dfn}\label{def master function}
We define 
the master functions  $\rho_\mathrm{J}$  as follows. 
\begin{align*}
&\rho_\mathrm{VI}(u,t,a,b,c,d)=u^{-a-b-1}(1-u)^{-c-1}(t-u)^{-d},
\\
&\rho_\mathrm{V}( u,t,b,c)=u^{-b-1}(1-u)^{-c-1}\exp(ut),
\\
&\rho_\mathrm{IV}(u,t,b)=u^{-b-1}\exp\left(-
ut+\frac{u^2}{2}
\right),
\\
&\rho_\mathrm{III}(u,t,b)=u^{-b-1}\exp\left(
\frac{t}{u}-u
\right),
\\
&\rho_\mathrm{II}(u,t)=\exp\left(
-\left(
ut+\frac{2}{3}u^3
\right)
\right).
\end{align*}
 
\end{dfn}



\begin{thm}\label{thm HS}
For $m\in\Z_{\ge 0}$, 
the integral formulas 
\begin{equation*}
\Phi^\mathrm{J}_m(x,t,{\bf a},\hbar)=\int_{\Gamma_m^J}\prod_{1\le i<j\le m}(u_i-u_j)^{2\hbar}\prod_{i=1}^m
\rho_\mathrm{J}(u_i,t,{\bf a})(x-u_i)du_i, 
\end{equation*} 
where $\Gamma_m^J$ is an $m$-cycle of the homology group determined by the integrand 
and
\begin{equation*}
{\bf a}=\left\{
\begin{tabular}{cc}
$(a,b,c,d)$  & $\mathrm{J=VI}$, 
\\ 
$ (b,c)$ & $\mathrm{J=V}$,  
\\
$ b$ & $\mathrm{J=III, IV}$, 
\end{tabular}
\right.
\end{equation*}
are solutions to   the Schr\"odinger equations for the 
quantum Painlev\'e equations of type  $\mathrm{J=II,III,IV,V,VI}$ with 
$a=m\hbar $ or of type $\mathrm{VI}$ with $b+c+d=(m-1)\hbar$.
 \end{thm}


 Let  singular loci $U_\mathrm{J}\subset\C^m$ 
($\mathrm{J=II,III,IV,V,VI}$) be defined as 
\begin{align*}
&U_\mathrm{VI}=\bigcup_{1\le i<j\le m} \{ u_i=u_j\}\cup \bigcup_{1\le i\le m}\{u_i=0\}\cup 
\bigcup_{1\le i\le m}\{u_i=1\}\cup \bigcup_{1\le i\le m}\{u_i=t\},
\\
&U_\mathrm{V}=\bigcup_{1\le i<j\le m} \{ u_i=u_j\}\cup \bigcup_{1\le i\le m}\{u_i=0\}\cup 
\bigcup_{1\le i\le m}\{u_i=1\},
\\
&U_\mathrm{IV}=U_\mathrm{III}=\bigcup_{1\le i<j\le m} \{ u_i=u_j\}\cup \bigcup_{1\le i\le m}\{u_i=0\},
\\
&U_\mathrm{II}=\bigcup_{1\le i<j\le m} \{ u_i=u_j\}.
\end{align*}

For a rational function $\varphi(u_1,\ldots,u_m)$ holomorphic outside  $U_\mathrm{J}$, denote by $\langle \varphi(u_1,\ldots,u_m) \rangle_m$ an integral formula
\begin{equation*}
\int_{\Gamma_m^J} \prod_{1\le i<j\le m}(u_i-u_j)^{2\hbar}\prod_{i=1}^m
\rho_\mathrm{J}(u_i, {\bf a}) \varphi(u_1,\ldots,u_m) du_i. 
\end{equation*}
Let $\nabla_i$ ($i=1,\ldots, m$) be the differential defined by 
\begin{equation*}
\nabla_i=\frac{\partial}{\partial u_i}+\frac{\partial}{\partial u_i}
\left(\log\left(\prod_{1\le i<j\le m}(u_i-u_j)^{2\hbar}\prod_{i=1}^m
\rho_\mathrm{J}(u_i,t,{\bf a})\right)\right). 
\end{equation*}
 By definition, it holds
 \begin{equation*}
\left\langle
\nabla_i\left(\varphi(u_1,\ldots, u_m)\right)\right\rangle
=\int_{\Gamma_m^J}\prod_{j=1}^m du_j \frac{\partial}{\partial u_i}
\left(\prod_{1\le i<j\le m}(u_i-u_j)^{2\hbar}\prod_{i=1}^m
\rho_\mathrm{J}(u_i, {\bf a}) \varphi(u_1,\ldots,u_m)\right)
=0.
\end{equation*}

Note that the polynomial $\prod_{i=1}^m (x-u_i)$ in the integral formula 
is symmetrical with respect to the integral variables $u_1, \ldots,u_m$. 
It is convenient to use the symmetrization for computations of the integral formula. 
Let  $\mathop{\mathrm{Sym}}\left[\varphi(u_1,\ldots, u_m)\right]$ be the symmetrization given by  
\begin{equation*}
\mathop{\mathrm{Sym}}\left[\varphi(u_1,\ldots, u_m)\right]
=\frac{1}{m!}
\sum_{\sigma\in \frak{S}_m}
\varphi(u_{\sigma(1)},\ldots,u_{\sigma(m)}), 
\end{equation*}
where $\frak{S}_m$ is the symmetric group of degree $m$.

\medskip

{\bf A proof of Theorem \ref{thm HS} for the case of $\mathrm{J=VI}$. }
Let $\Phi^\mathrm{VI}_m(x,t,a,b,c,d,\hbar)$ be written as
\begin{equation*}
\Phi^\mathrm{VI}_m(x,t,a,b,c,d,\hbar)=\sum_{k=0}^m (-1)^{k+1}\begin{pmatrix}
m\\ k
\end{pmatrix}\left\langle\mathop{\mathrm{Sym}}\left[ \prod_{l=1}^{m-k} u_l \right]\right\rangle_m x^k, 
\end{equation*}
where $\begin{pmatrix}
m\\ k
\end{pmatrix}$ is the binomial coefficient. Denote $ \mathop{\mathrm{Sym}}\left[\prod_{l=1}^{m-k} u_l\right]$ by $\varphi_k$.   
 For  $ k=0,1,\ldots, m$, we have 
\begin{equation*}
\frac{\partial}{\partial t}\langle \varphi_k\rangle_m = 
(k-m)d\left\langle \mathop{\mathrm{Sym}}\left[\left(\frac{t}{t-u_{m-k}}-1\right)
\prod_{l=1}^{m-k-1}u_l
\right]\right\rangle_m-kd\left\langle \mathop{\mathrm{Sym}}\left[\frac{1}{t-u_{m-k+1}}\prod_{l=1}^{m-k}u_l\right]\right\rangle_m. 
\end{equation*}
Applying Lemma \ref{lem HS} below, we obtain 
\begin{align}
 t(t-1)\frac{\partial}{\partial t}\langle \varphi_k\rangle_m =& (m-k) t \left(
a+b-k\hbar
\right)\left\langle
\varphi_{k+1}
\right\rangle_m \label{eq varphi}
 \\
&  -\left\{(m-k) t (a+b+c+d  -(m+k-1)\hbar \right.\nonumber
\\
&\left.-k(a+b+d(1-t)-(k-1)\hbar))
\right\}\left\langle
\varphi_k
\right\rangle_m \nonumber
\\
&-k(a+b+c+d -(m+k-2)\hbar)\left\langle
\varphi_{k-1}
\right\rangle_m. \nonumber 
\end{align}
On the other hand, the action of the Hamiltonian $\widehat{H}_\mathrm{{VI}}$ on $\Phi^\mathrm{VI}_m(x,t,a,b,c,d,\hbar)$ 
is easily calculated and we see that 
the coefficient of $x^k$ of 
$$  
t(t-1)\frac{\widehat{H}_\mathrm{{VI}}}{\hbar } \Phi^\mathrm{VI}_m(x,t,a,b,c,d,\hbar)
$$
times $(-1)^{k+1}\begin{pmatrix}
m\\ k
\end{pmatrix}^{-1}$ coincides with the right hand side of \eqref{eq varphi}, which finishes the proof. \qed


\begin{lem}\label{lem HS}
For, $k=1,\ldots, m$, we have
\begin{align}
 t(t-1)\left\langle\mathop{\mathrm{Sym}}\left[
\frac{d}{t-u_{m-k}}\prod_{l=1}^{m-k-1}u_l\right]
\right\rangle_m
=&-\left(
a+b+d (1-t)-k\hbar
\right)\left\langle
\varphi_{k+1}
\right\rangle_m \label{eq lem HS}
\\ &
+\left(
a+b+c+d -(m+k-1)\hbar
\right)\left\langle
\varphi_k
\right\rangle_m. \nonumber
\end{align}
\end{lem}
\begin{proof}
In order to prove \eqref{eq lem HS}, we compute 
\begin{equation*}
X=
\left\langle \mathop{\mathrm{Sym}}\left[\nabla_i\left(u_1(1-u_1)\prod_{l=2}^{m-k} u_l\right)\right]\right\rangle
=0
\end{equation*}
as follows. From the definition, we get 
\begin{align*}
X=&-(a+b+d(1-t))\left\langle
\varphi_{k+1}
\right\rangle_m
+(a+b+c+d )\left\langle
\varphi_k
\right\rangle_m
- t(t-1)\left\langle\mathop{\mathrm{Sym}}\left[
\frac{d}{t-u_{m-k}}\prod_{l=1}^{m-k-1}u_l\right]
\right\rangle_m
 \\
 &+2\hbar\left\langle\mathop{\mathrm{Sym}}\left[
\sum_{l=2}^m\frac{u_1(1-u_1)}{u_1-u_l}u_2\cdots u_{m-k}\right]
\right\rangle_m. 
\end{align*}
Since the symmetrization is invariant 
 under the action of $\frak{S}_m$ on the integral variables $u_1,\ldots, u_m$, we have
 \begin{align*}
2\left\langle\mathop{\mathrm{Sym}}\left[
\sum_{l=2}^m\frac{u_1(1-u_1)}{u_1-u_l}u_2\cdots u_{m-k}\right]
\right\rangle_m= & \left\langle\mathop{\mathrm{Sym}}\left[
\sum_{l=2}^m\frac{u_1(1-u_1)}{u_1-u_l}u_2\cdots u_{m-k}\right]
\right\rangle_m
\\
&
-\left\langle\mathop{\mathrm{Sym}}\left[
\sum_{l=2}^{m-k}\frac{u_l(1-u_l)}{u_1-u_l}u_1u_2\cdots u_{l-1}u_{l+1}\cdots u_{m-k}\right]
\right\rangle_m 
\\
&-\left\langle\mathop{\mathrm{Sym}}\left[
\sum_{l=m-k+1}^m\frac{u_l(1-u_l)}{u_1-u_l}u_2\cdots u_{m-k}\right]
\right\rangle_m
\\
=&k\left\langle
\varphi_{k+1}
\right\rangle_m 
-(m+k-1)\left\langle
\varphi_k
\right\rangle_m.
\end{align*}
  Therefore, we arrive at \eqref{eq lem HS}. 
\end{proof}

For the other cases, we can verify Theorem \ref{thm HS} in the similar way. 
It also follows 
from Theorem \ref{thm ckz v}, Theorem \ref{thm ckz iv},  
Theorem \ref{thm ckz iii}, Theorem \ref{thm ckz ii}  in section 3,   and Proposition 4.2 in \cite{JNS}. 
We give an another proof of Theorem \ref{thm HS} for the case of $\mathrm{J=V}$  in section 3.

As mentioned in Introduction, 
if $\hbar=1$ and $\Gamma_m^J=\prod_{i=1}^m \Gamma_1^J$, where $\Gamma^J_1$ is an appropriate cycle for 
$\Phi_1^\mathrm{J}(x,t)$, then the integral formulas  $\Phi_m^\mathrm{J}(x,t)$ ($m\in\Z_{\ge 0}$) are expressed by the 
determinant formulas $P_m^\mathrm{J}(x,t)$ \eqref{eq determinant formula} and $P_m^\mathrm{J}(x,t)$ are 
orthogonal polynomials 
in $x$ with respect to the weight functions $\rho_\mathrm{J}(x)$. Namely, it holds that 
\begin{equation*}
\int_{\Gamma^J_1} P_m^\mathrm{J}(x,t)P_n^\mathrm{J}(x,t)\rho_J(x)dx =0. 
\end{equation*}

From Theorem \ref{thm HS}, we have 

\begin{cor}
If $\hbar=1$, then the determinant formulas $P_m^J(x,t)$ ($\mathrm{J=II, III, IV, V, VI}$) are  solutions to 
  the Schr\"odinger equations for the 
quantum  Painlev\'e equations of type $\mathrm{II, III, IV, V, VI}$ with $a=m$ or of type $\mathrm{VI}$ with 
$b+c+d=m-1$. 
\end{cor}


\section{Relation to the KZ equation}\label{sec kz}
In this section, we give a representation theoretic correspondence between  the Schr\"odinger equations for the sixth quantum Painlev\'e equation 
and the  KZ equation, and between  the Schr\"odinger equations for the other quantum Painlev\'e equations and the confluent KZ equations, 
which  were defined directly in \cite{JNS} and can be derived from the irregular conformal field theory \cite{NS}. 
In what follows, we only give necessary facts 
on confluent KZ equations 
in this paper. See \cite{JNS} for the detail.

Let us recall the definition of confluent Verma modules in \cite{JNS}. Set $\g=\slt$ and $\g[z]=\g\otimes \C[z]$. 
Denote by $e,f,h$ the standard basis of $\g$. 
For non-negative integer $r$, 
denote by $\g_{(r)}$ and 
$\g'_{(r)}$ 
the truncated Lie algebras $\g[z]/z^{r+1}\g[z]$ and $\g'_{(r)}=z\g[z]/z^{r+1}\g[z]$, respectively. 
For an $(r+1)$-tuple parameters 
$\gamma =(\gamma_0,\ldots,\gamma_{r-1},\gamma_r)\in\C^r\times\C^\times$, 
a confluent Verma 
module $M(\gamma)$ of Poincar\'e rank $r$   
is a  cyclic $\g_{(r)}$-module generated by ${\bf 1}_\gamma$ such that 
\begin{equation*}
(e\otimes z^p){\bf 1}_\gamma=0, \quad (h\otimes z^p){\bf 1}_\gamma=\gamma_p{\bf 1}_\gamma 
\quad (0\le p\le r). 
\end{equation*}
For the Lie subalgebra $\g'_{(r)}=z\g[z]/z^{r+1}\g[z]$, a confluent Verma module $M'(\gamma)$ 
of Poincar\'e rank $r$  with parameters $\gamma=(\gamma_1,\ldots,\gamma_r)\in\C^{r-1}\times \C^\times$ 
is a cyclic $\g'_{(r)}\oplus \C (h\otimes z^0)$-module generated by ${\bf 1}_\gamma$ such that 
\begin{equation*}
(e\otimes z^p){\bf 1}_\gamma=0, \quad (h\otimes z^p){\bf 1}_\gamma=\gamma_p{\bf 1}_\gamma 
\quad (1\le p\le r), \quad 
(h\otimes z^0){\bf 1}_\gamma=0, 
\end{equation*}
and $e\otimes z^r$ and $f\otimes z^r$ act as zero operators on $M'(\gamma)$. 

Let differential operators $D_k$ ($ 0\le k\le r-1$) be defined as 
\begin{equation*}
D_k=\sum_{p=1}^{r-k}p\gamma_{k+p}\frac{\partial}{\partial \gamma_p} 
\end{equation*}
acting on $M(\gamma)$ as 
\begin{equation*}
D_k(x\otimes z^p)=p(x\otimes z^{p+k})\quad (x\in\g, 0\le p\le r), \quad D_k({\bf 1}_\gamma)=0. 
\end{equation*}
Here we regard $x\otimes z^p$ as  an operator on $M(\gamma)$.   
\begin{exmp}
For the $\g'_{(2)}$ case, $M'(\gamma_1,\gamma_2)$ is realized on the polynomial ring $\C[x]$. 
The action of $x\otimes z^p$ ($x\in\g$, $p=1,2$) are following. 
\begin{equation*}
e\otimes z=\gamma_2^{\frac{1}{2}}\frac{\partial}{\partial x},\quad e\otimes z^2=0,\quad 
f\otimes z=\gamma_2^{\frac{1}{2}}x,\quad f\otimes z^2=0, \quad h\otimes z=\gamma_1,\quad h\otimes z^2=\gamma_2.
\end{equation*}
\end{exmp}

Let $z_1, \ldots, z_n$ be distinct points in $\mathbb{C}$ and let $r_1,\ldots, r_n, r_\infty$ be non-negative integers. 
Set $\frak{a}=\left(
\oplus_{i=1}^n\g^{(i)}
\right)\oplus \g^{(\infty)}$, where $\g^{(i)}=\g_{(r_i)}$ ($i=1,\ldots, n$) and $\g^{(\infty)}=\g'_{(r_\infty)}$. 
We consider a family of $\frak{a}$-modules 
\begin{equation*}
M(\gamma )=M^{(1)}\otimes \cdots \otimes M^{(n)}\otimes M^{(\infty)}, 
\end{equation*}
parametrized by $\gamma=\left(
\gamma^{(1)},\ldots, \gamma^{(n)},\gamma^{(\infty)}
\right)$, where 
\begin{align*}
&M^{(i)}=M(\gamma^{(i)}), \quad \gamma^{(i)}=\left(\gamma_0^{(i)},\ldots, \gamma_{r_i}^{(i)}\right)
\in \C^{r_i}\times \C^\times,
\\
&M^{(\infty)}=M'(\gamma^{(\infty)}),\quad \gamma^{(\infty)}=\left(\gamma_1^{(\infty)},\ldots, \gamma_{r_\infty}^{(\infty)}\right)
\in \C^{r_\infty-1}\times \C^\times.
\end{align*}
Set ${\bf 1_\gamma=1_{\gamma^{(1)}} \otimes \cdots \otimes 1_{\gamma^{(n)}} \otimes 1_{\gamma^{(\infty)}}}$. 

The confluent KZ equations defined in \cite{JNS} are differential systems for unknown functions  $\Phi(z,\gamma)$ taking values in $M(\gamma)$ 
with respect to the following differential operators 
\begin{align*}
&\frac{\partial}{\partial z_i}\quad (i=1,\ldots, n), \quad  
\\&D^{(i)}_k\quad (i=1,\ldots,n, k=0,\ldots, r_i-1), \quad 
\\&D^{(\infty)}_k\quad ( k=1,\ldots, r_\infty-1).  \quad 
\end{align*}
If $r_i=0$ ($1\le i\le n$) and $r_\infty=0$, then the confluent KZ equations are equal to the usual KZ equations. 

It was shown  in \cite{JNS} that the confluent KZ equations have integral formulas of 
confluent hypergeometric type for  solutions.

\subsection{The case of $\mathrm{J=VI}$} 

Let $n=3$, $r_i=0$, $z_1=0$, $z_2=t$, $z_3=1$ and $\gamma^{(i)}_0\not\in\Z$ ($1\le i\le 3$). 
Then $M=M(\gamma^{(1)}_0)\otimes M(\gamma^{(2)}_0)\otimes M(\gamma^{(3)}_0)$ 
and the KZ equation for 
an unknown function $\Phi(t)$ taking values in $M$ is 
\begin{equation}\label{eq kz}
\kappa \frac{\partial \Psi(t)}{\partial t}=\left(
\frac{\Omega^{(1,2)}}{t}+\frac{\Omega^{(2,3)}}{t-1}
\right)\Psi(t).
\end{equation}
Here $\kappa$ is a complex parameter and $\Omega^{(i,j)}$ 
are the Casimir operators: 
\begin{equation*}
\Omega^{(1,2)}=e^{(1)}f^{(2)}+f^{(1)}e^{(2)}+\frac{1}{2}h^{(1)}h^{(2)},
\quad 
\Omega^{(2,3)}=e^{(2)}f^{(3)}+f^{(2)}e^{(3)}+\frac{1}{2}h^{(2)}h^{(3)}, 
\end{equation*}
where  $x^{(i)}:M\to M$ is the linear operator acting as $x$ on $i$th tensor factor and 
as identities on the others, and 
here and after,  we abbreviate $x\otimes z^0$ to $x$, for $x=e,f,h$.

Let $W_m$ ($m\in\Z_{\ge 0}$) be the space of singular vectors of  
the weight $\sum_{i=1}^3\gamma^{(i)}_0-2m$ in $M$, namely, 
\begin{equation*}
W_m=\left\{
v\in M \bigg| \sum_{i=1}^3 e^{(i)} (v)=0, \ \sum_{i=1}^3h^{(i)}(v)=
\left(\sum_{i=1}^3\gamma^{(i)}_0-2m\right)v
\right\}. 
\end{equation*}
If $\gamma^{(i)}_0\not\in\Z$, then 
the dimension of $W_m$ is known to be $m+1$ (for example, see Proposition 4.1.1 in \cite{KZ EFK}). 

In order to write down a basis of $W_m$, we take the differential realizations $\C[x_i]$ ($1\le i\le 3$) of $\g$, that is, the basis $e,f,h$ of $\g$ act on $\C[x_i]$ 
as follows: 
\begin{equation*}
e=\partial_i,
\quad
 h=-2x_i\partial_i+\gamma^{(i)}_0, 
 \quad
 f=-x_i^2\partial_i + \gamma^{(i)}_0x_i, 
\end{equation*}
where $\partial_i=\partial/\partial x_i$.
Note that if $\gamma^{(i)}_0\not\in\Z$, then $\C[x_i]$ are  
isomorphic to Verma modules $M(\gamma^{(i)}_0)$. We set $M(\gamma^{(i)}_0)=\C[x_i]$.

The space of singular vectors $W_m$  can be written by 
\begin{equation}\label{eq singular vector space}
W_m=\bigoplus_{i=0}^m\C(x_1-x_2)^i(x_1-x_3)^{m-i}.  
\end{equation}
We denote by $H_\mathrm{{KZ}}$ the Hamiltonian 
$\Omega^{(1,2)}/t+\Omega^{(2,3)}/(t-1)$. 
Let $\widetilde{H}_\mathrm{{KZ}}(m)$ be defined as 
\begin{equation*}
\widetilde{H}_\mathrm{{KZ}}(m)=\hbar^2\left(H_\mathrm{{KZ}}
-\frac{\lambda_1\lambda_2}{t}-\frac{\lambda_2(\lambda_3-m)}{t-1}
\right). 
\end{equation*}
We define  linear isomorphisms $T_m: W_m\to \bigoplus_{i=0}^m\C x^i$  ($m\in\Z_{\ge 0}$) as 
\begin{equation*}
T_m\left((x_1-x_2)^i(x_1-x_3)^{m-i}\right)=x^i\quad (0\le i\le m). 
\end{equation*}

\begin{thm}\label{thm kz} For $\gamma^{(i)}_0\not\in\Z$ ($1\le i\le 3$) and $m\in\mathbb{Z}_{\ge 0}$,  
the action of $H_\mathrm{{KZ}}$ 
on  the space of singular vectors of weight $\sum_{i=1}^3 \gamma^{(i)}_0-2m$ is 
equivalent to the action of the quantized Hamiltonian $\widehat{H}_\mathrm{{VI}}$ on the subspace 
$\oplus_{i=0}^m\mathbb{C}x^i$  with $a=m\hbar$  or $b+c+d=(m-1)\hbar$.    
Namely,  we have 
\begin{equation}\label{eq hamiltonian action}
T_m\circ \widetilde{H}_\mathrm{{KZ}}(m) = \left(\widehat{H}_\mathrm{{VI}}\left(
x,\hbar\frac{\partial}{\partial x}, a, b, c,d,t
\right)
+\frac{a(b+c+d+\hbar)}{t-1}\right)
\circ T_m
\end{equation}
as linear maps from $W_m$ to $\bigoplus_{i=0}^m\C x^i$
with 
\begin{align}
&
\gamma^{(1)}_0=m-1-\frac{c}{\hbar},\label{eq kz 1} 
\\
&
\gamma^{(2)}_0=\frac{1}{\hbar}\left(
a+b+c+d+(1-m)\hbar
\right),
\quad
\gamma^{(3)}_0=m-1-\frac{a+b}{\hbar}, \label{eq kz 2}
\end{align}
 and 
 \begin{equation}\label{eq qp6 m}
a=m\hbar \quad   \text{or} \quad  b+c+d=(m-1)\hbar.  
\end{equation}
  
\end{thm}
\begin{proof}
It follows from direct computations. 
\end{proof}

We note that the condition $\gamma^{(i)}_0\not\in\Z$ ensures $M(\gamma^{(i)}_0)=\C[x_i]$, and 
for any $\gamma^{(i)}_0\in\C$ ($1\le i\le 3$), it holds \eqref{eq hamiltonian action} as linear maps from  
$\bigoplus_{i=0}^m\C(x_1-x_2)^i(x_1-x_3)^{m-i}$ to $\bigoplus_{i=0}^m\C x^i$ 
with the conditions \eqref{eq kz 1}, \eqref{eq kz 2}, and \eqref{eq qp6 m}. 
 

From Theorem \ref{thm kz}, we obtain a solution $\Psi(t)$ to the KZ equation \eqref{eq kz} 
from the 
integral formula $\Phi^{\mathrm{VI}}_m(x, t)$ for a solution to the Schr\"odinger equation for the sixth 
quantum Painlev\'e equation, that is, 
$$
\Psi(t)=t^{(\gamma^{(1)}_0\gamma^{(2)}_0)/\kappa}(t-1)^{(\gamma^{(2)}_0\gamma^{(3)}_0)/\kappa}T_m^{-1}(\Phi^{\mathrm{VI}}_m(x, t))
$$
obtained by 
replacing 
 $\hbar$, $a+b$, $c$, $d$ with 
\begin{equation*}
\frac{1}{\ka}, \quad 
a+b=\hbar(m-1-\gamma^{(3)}_0), \quad c=\hbar(m-1-\gamma^{(1)}_0), \quad 
d=\hbar(\gamma^{(1)}_0+\gamma^{(2)}_0+\gamma^{(3)}_0+1-m), 
\end{equation*}
respectively. 
Conversely, from Theorem \ref{thm kz}, we obtain 
\begin{equation*}
\left(
\hbar^2 \ka \frac{\partial}{\partial t}-\widehat{H}_\mathrm{VI}
\right)T_m\left(
\widetilde{\Psi}(t)
\right)=0,  
\end{equation*}
where $\widetilde{\Psi}(t)=t^{-(\gamma^{(1)}_0\gamma^{(2)}_0)/\kappa}(t-1)^{-(\gamma^{(2)}_0\gamma^{(3)}_0)/\kappa}
\Psi(t)$. 
If $\ka=1/\hbar$, then $T_m(\widetilde{\Psi}(t))=\Phi^{\mathrm{VI}}_m(x, t)$. 
Taking $\kappa=-1/\hbar$, we have 
\begin{equation*}
\left(
\hbar \frac{\partial}{\partial t}+\widehat{H}_\mathrm{IV}
\right)\Phi^\mathrm{VI}_m(x,t, -a,-b,-c,-d,-\hbar)=0. 
\end{equation*}
This implies the following symmetry. 
\begin{prop}
If a function $\Phi^{\mathrm{VI}}(x,t,a,b,c,d,\hbar)$ is a solution to the Schr\"odinger equation for 
the quantum sixth Painlev\'e equation , then 
the function $\Psi^{\mathrm{VI}}(x,t,a,b,c,d,\hbar)$ defined as 
\begin{equation}\label{def h 6}
\Psi^{\mathrm{VI}}(x,t,a,b,c,d,\hbar)=\Phi^{\mathrm{VI}}(x,t,-a,-b,-c,-d,\hbar)
\end{equation}
is a solution to 
\begin{equation*}
\hbar\frac{\partial}{\partial t}\Psi^\mathrm{VI}(x,t)=\widehat{H'}_\mathrm{VI} \Psi^\mathrm{VI}(x,t), 
\end{equation*}
where $\widehat{H'}_\mathrm{VI}$ are obtained by replacing $\hbar$ with $-\hbar$ in $\widehat{H}_\mathrm{VI}$. 
\end{prop}
\begin{proof}
 It follows from direct computations. 
\end{proof}

\subsection{The case of $\mathrm{J=V}$}
Let $n=2$, $r_1=r_2=0$, $r_\infty=1$, $z_1=0$, and $z_2=1$. Then $M=M(\gamma^{(1)}_0)\otimes M(\gamma^{(2)}_0)\otimes M'(\gamma^{(\infty)}_1)$ 
and the confluent KZ equation for an unknown function $\Phi(\gamma^{(\infty)}_1)$ taking values in $M$ is 
\begin{equation}\label{eq 5}
\kappa \gamma^{(\infty)}_1\frac{\partial}{\partial \gamma^{(\infty)}_1} \Psi(\gamma^{(\infty)}_1)=
\left(
G_0^{(\infty)}-\frac{1}{4}h_0(h_0+2)
\right)\Psi(\gamma^{(\infty)}_1), 
\end{equation}
where $\kappa\in\C$, $G_0^{(\infty)}=G_{-1}^{(2)}+G_0^{(1)}+G_0^{(2)}$, 
\begin{align*}
G_{-1}^{(2)}=&-\frac{1}{2}h^{(2)}\gamma^{(\infty)}_1+e^{(1)}f^{(2)}+f^{(1)}e^{(2)}+\frac{1}{2}h^{(1)}h^{(2)},
\\
G_0^{(i)}=&\frac{1}{2}\left(
e^{(i)}f^{(i)}+f^{(i)}e^{(i)}+\frac{1}{2}h^{(i)}h^{(i)}
\right)\quad (i=1,2),
\\
h_0=&h^{(1)}+h^{(2)}. 
\end{align*}

Let 
$M_m$ ($m\in\Z_{\ge 0}$) be the weight space of $M$ with the weight $\gamma^{(1)}_0+\gamma^{(2)}_0-2m$, namely, 
\begin{align*}
M_m=&\left\{
v\in M\ \bigg|  \ h_0(v)=
\left(\gamma^{(1)}_0+\gamma^{(2)}_0-2m\right)v
\right\}
\\=&\bigoplus_{i=0}^m\C (f^{(1)})^i(f^{(2)})^{m-i}{\bf 1_\gamma}. 
\end{align*}
An integral formula taking values in $M_m$ for solutions to the confluent KZ equation \eqref{eq 5}  is 
\begin{align}
\int_\Gamma &\prod_{i=1}^mdu_i\exp\left(
-\frac{1}{2\kappa}\gamma^{(2)}_0\gamma^{(\infty)}_1
\right)\prod_{1\le i<j\le m}(u_i-u_j)^{2/\kappa} \prod_{i=1}^mu_i^{-\gamma^{(1)}_0/\ka}(1-u_i)^{-\gamma^{(2)}_0/\ka} 
\label{eq integral formula v}
\\
&\times \exp\left(
\frac{1}{\ka}\sum_{i=1}^mu_i\gamma^{(\infty)}_1
\right)
\prod_{i=1}^m\left(
\frac{f^{(1)}}{u_i}+\frac{f^{(2)}}{u_i-1}
\right){\bf 1_\gamma}, \nonumber
\end{align}
where $\Gamma$ is an appropriate cycle (see Proposition 4.2 in \cite{JNS}).

We define linear isomorphisms $T_m: M_m\to \bigoplus_{i=0}^m\C x^i$ ($m\in\Z_{\ge 0}$) as 
\begin{equation*}
T_m\left((f^{(1)})^i
(f^{(1)}+f^{(2)})^{m-i}{\bf 1_\gamma}
\right)
=x^i\quad (0\le i\le m).
\end{equation*}

\begin{thm}\label{thm ckz v}
For $m\in\Z_{\ge 0}$, we have 
\begin{equation*}
T_m\circ \left(
G_0^{(\infty)}-\frac{1}{4}h_0(h_0+2)+\frac{1}{2}\gamma^{(2)}_0\gamma^{(\infty)}_1
\right)=\left(\frac{t}{\hbar^2}\widehat{H}_\mathrm{V}\left(
x, \hbar\frac{\partial}{\partial x}, m\hbar, b, c, t
\right)\right)\circ T_m
\end{equation*}
as linear maps from $M_m$ to $\bigoplus_{i=0}^m\C x^i$ 
with 
\begin{equation}\label{eq condition ckz v}
\gamma^{(1)}_0=\frac{b}{\hbar}, \quad \gamma^{(2)}_0=\frac{c}{\hbar}, \quad \gamma^{(\infty)}_1=\frac{t}{\hbar}.
\end{equation} 
\end{thm}
\begin{proof}
It follows from direct computations. 
\end{proof}
{\bf A proof of Theorem \ref{thm HS} for the  case of $\mathrm{J=V}$.} 
Substitute \eqref{eq condition ckz v} into \eqref{eq integral formula v}. Then we have
\begin{equation*}
\left(
\kappa t \frac{\partial}{\partial t}-G_0^{(\infty)}+\frac{1}{4}h_0(h_0+2)-\frac{ct}{2\hbar^2}
\right)\Phi(t)=0, 
\end{equation*}
where 
\begin{align}
\Phi(t)=\int_\Gamma &\prod_{i=1}^mdu_i
\prod_{1\le i<j\le m}(u_i-u_j)^{2/\kappa} \prod_{i=1}^mu_i^{-b/(\hbar\ka)-1}(1-u_i)^{-c/(\hbar\ka)-1} 
\label{eq integral formula v 1}
\\
&\times \exp\left(
\frac{1}{\hbar\ka}\sum_{i=1}^mu_i t
\right)
\prod_{i=1}^m\left(
f^{(1)}-(f^{(1)}+f^{(2)})u_i
\right){\bf 1_\gamma}.  \nonumber
\end{align}
From Theorem \ref{thm ckz v}, we obtain 
\begin{equation}\label{eq ckz v 1}
\left(
\kappa t\frac{\partial}{\partial t}-\frac{t}{\hbar^2}\widehat{H}_\mathrm{V}
\right)T_m\left(
\Phi(t)
\right)=0.
\end{equation}
If  $\ka=1/\hbar$, then \eqref{eq ckz v 1} becomes 
\begin{equation*}
\left(
\hbar \frac{\partial}{\partial t}-\widehat{H}_\mathrm{V}
\left(
x, \hbar\frac{\partial}{\partial x}, m\hbar, b, c, t
\right)\right)
\Phi_m^\mathrm{V}(x,t,b,c,\hbar)=0. 
\end{equation*}
\qed

From \eqref{eq ckz v 1},  an equation
\begin{equation*}
\hbar^2 \kappa \frac{\partial}{\partial t}\Phi(x,t)=\widehat{H}_\mathrm{V}
\left(
x, \hbar\frac{\partial}{\partial x}, m\hbar, b, c, t
\right)\Phi(x,t)
\end{equation*}
also has integral formulas for solutions. 
 Especially, taking $\kappa=-1/\hbar$, we have 
\begin{equation*}
\left(
\hbar \frac{\partial}{\partial t}+\widehat{H}_\mathrm{V}
\left(
x, \hbar\frac{\partial}{\partial x}, m\hbar, b, c, t
\right)\right)
\Phi_m^\mathrm{V}(x,-t,-b,-c,-\hbar)=0
\end{equation*}
from \eqref{eq integral formula v 1}. 
 This implies the following 
symmetry. 
\begin{prop}
If a function $\Phi^{\mathrm{V}}(x,t,a,b,c,\hbar)$ is a solution to the Schr\"odinger equation for 
the quantum fifth Painlev\'e equation, then 
the function $\Psi^{\mathrm{V}}(x,t,a,b,c,\hbar)$ defined as 
\begin{equation}\label{def h 5}
\Psi^{\mathrm{V}}(x,t,a,b,c,\hbar)=\Phi^{\mathrm{V}}(x,-t,-a,-b,-c,\hbar)
\end{equation}
is a solution to 
\begin{equation*}
\hbar\frac{\partial}{\partial t}\Psi^\mathrm{V}(x,t)=\widehat{H'}_\mathrm{V} \Psi^\mathrm{V}(x,t), 
\end{equation*}
where $\widehat{H'}_\mathrm{V}$ are obtained by replacing $\hbar$ with $-\hbar$ in $\widehat{H}_\mathrm{V}$. 
\end{prop}
\begin{proof}
 It follows from direct computations. 
\end{proof}



\subsection{The case of $\mathrm{J=IV}$}
Let $n=1$, $r_1=0$, $r_\infty=2$, and $z_1=0$. Then $M=M(\gamma^{(1)}_0)\otimes M'(\gamma^{(\infty)}_1,\gamma^{(\infty)}_2)$ 
and the confluent KZ equation for an unknown function $\Psi(\gamma^{(\infty)}_1)$ taking values in $M$ is 
\begin{equation}\label{eq 4}
\kappa \gamma^{(\infty)}_2\frac{\partial}{\partial \gamma^{(\infty)}_1} \Psi(\gamma^{(\infty)}_1)=
\left(
G_1^{(\infty)}+\frac{1}{2}h_0\gamma^{(\infty)}_1
\right)\Psi(\gamma^{(\infty)}_1), 
\end{equation}
where $\kappa\in\C$, 
\begin{align*}
G_{1}^{(\infty)}=&-e^{(1)}(f\otimes z)^{(\infty)}+f^{(1)}(e\otimes z)^{(\infty)}+\frac{1}{2}h^{(1)}\gamma^{(\infty)}_1,
\\
h_0=&h^{(1)}+h^{(\infty)}. 
\end{align*}
Here $x^{(\infty)}$ for $x\in\g'_{(2)}\oplus \C (h\otimes 1)$ stands for the linear operator acting as $x$ on $M'(\gamma^{(\infty)}_1,\gamma^{(\infty)}_2)$ 
and as identity on the other. 

Let 
$M_m$ ($m\in\Z_{\ge 0}$) be the weight space of $M$ with the weight $\gamma^{(1)}_0-2m$, namely, 
\begin{align*}
M_m=&\left\{
v\in M\ \bigg|  \ h_0(v)=
\left(\gamma^{(1)}_0-2m\right)v
\right\}
\\=&\bigoplus_{i=0}^m\C (f^{(1)})^i((f\otimes z)^{(\infty)})^{m-i}{\bf 1_\gamma}. 
\end{align*}
An integral formula taking values in $M_m$ for solutions to the confluent KZ equation \eqref{eq 4}  is 
\begin{align}
\int_\Gamma &\prod_{i=1}^mdu_i
\prod_{1\le i<j\le m}(u_i-u_j)^{2/\kappa} \prod_{i=1}^mu_i^{-\gamma^{(1)}_0/\ka}
\label{eq integral formula iv}
\\
&\times \exp\left(
\frac{1}{\ka}\sum_{i=1}^m\left(u_i\gamma^{(\infty)}_1+\frac{u_i^2}{2}\gamma^{(\infty)}_2\right)
\right)
\prod_{i=1}^m\left(
\frac{f^{(1)}}{u_i}-(f\otimes z)^{(\infty)}
\right){\bf 1_\gamma}, \nonumber
\end{align}
where $\Gamma$ is an appropriate cycle (see Proposition 4.2 in \cite{JNS}). 

We define linear isomorphisms $T_m: M_m\to \bigoplus_{i=0}^m\C x^i$ ($m\in\Z_{\ge 0}$) as 
\begin{equation*}
T_m\left(
(f^{(1)})^{i}\left((f\otimes z)^{(\infty)}\right)^{m-i}{\bf 1_\gamma}
\right)
=x^i\quad (0\le i\le m).
\end{equation*}

\begin{thm}\label{thm ckz iv}
For $m\in\Z_{\ge 0}$, we have 
\begin{equation*}
T_m\circ \left(
G_1^{(\infty)}+\frac{1}{2}h_0\gamma^{(\infty)}_1
\right)=\frac{1}{\hbar^2}\widehat{H}_\mathrm{IV}
\left(
x, \hbar\frac{\partial}{\partial x}, m\hbar,b,t
\right)
\circ T_m
\end{equation*}
as linear maps from $M_m$ to $\bigoplus_{i=0}^m\C x^i$ 
with 
\begin{equation*}
\gamma^{(1)}_0=\frac{b}{\hbar},\quad  \gamma^{(\infty)}_1=-\frac{t}{\hbar}, \quad \gamma^{(\infty)}_2=\frac{1}{\hbar}. 
\end{equation*}
\end{thm}
\begin{proof}
It follows from direct computations. 
\end{proof}

As in the cases of $\mathrm{J=V, VI}$, Theorem \ref{thm ckz iv} implies the following 
symmetry. 
\begin{prop}
If a function $\Phi^{\mathrm{IV}}(x,t,a,b,\hbar)$ is a solution to the Schr\"odinger equation for 
the quantum fourth Painlev\'e equation, then 
the function $\Psi^{\mathrm{IV}}(x,t,a,b,\hbar)$ defined as 
\begin{equation}\label{def h 4}
\Psi^{\mathrm{IV}}(x,t,a,b,\hbar)=\Phi^{\mathrm{IV}}(\sqrt{-1}x, \sqrt{-1}t, -a,-b,\hbar)
\end{equation}
is a solution to 
\begin{equation*}
\hbar\frac{\partial}{\partial t}\Psi^\mathrm{IV}(x,t)=\widehat{H'}_\mathrm{IV} \Psi^\mathrm{IV}(x,t), 
\end{equation*}
where $\widehat{H'}_\mathrm{IV}$ are obtained by replacing $\hbar$ with $-\hbar$ in $\widehat{H}_\mathrm{IV}$. 
\end{prop}
\begin{proof}
 It follows from direct computations. 
\end{proof}

\subsection{The case of $\mathrm{J=III}$}
Let $n=1$, $r_1=2$, $r_\infty=1$, and $z_1=0$. Then $M=M(\gamma^{(1)}_0, \gamma^{(1)}_1)\otimes M'(\gamma^{(\infty)}_1)$ and 
the confluent KZ equation for an unknown function $\Psi(\gamma^{(\infty)}_1)$ taking values in $M$ is 
\begin{equation}\label{eq 3}
\kappa \gamma^{(\infty)}_1\frac{\partial}{\partial \gamma^{(\infty)}_1} \Psi(\gamma^{(\infty)}_1)=
\left(
G_0^{(\infty)}-\frac{1}{4}h_0(h_0+2)
\right)\Psi(\gamma^{(\infty)}_1), 
\end{equation}
where $\kappa\in\C$, 
\begin{align*}
G_{0}^{(\infty)}=&\frac{1}{2}\left(e^{(1)}f^{(1)}+f^{(1)}e^{(1)}+\frac{1}{2}h^{(1)}h^{(1)}\right)
-\frac{1}{2}(h\otimes z)^{(1)}\gamma^{(\infty)}_1,
\\
h_0=&h^{(1)}+h^{(\infty)}. 
\end{align*}

Let 
$M_m$ ($m\in\Z_{\ge 0}$) be the weight space of $M$ with the weight $\gamma^{(1)}_0-2m$, namely, 
\begin{align*}
M_m=&\left\{
v\in M\ \bigg|  \ h_0(v)=
\left(\gamma^{(1)}_0-2m\right)v
\right\}
\\=&\bigoplus_{i=0}^m\C (f^{(1)})^i((f\otimes z)^{(1)})^{m-i}{\bf 1_\gamma}. 
\end{align*}
An integral formula taking values in $M_m$ for a solution to the confluent KZ equation \eqref{eq 3}  is 
\begin{align}
\int_\Gamma &\prod_{i=1}^mdu_i
\prod_{1\le i<j\le m}(u_i-u_j)^{2/\kappa} \prod_{i=1}^mu_i^{-\gamma^{(1)}_0/\ka}
\label{eq integral formula iii}
\\
&\times \exp\left(
\frac{1}{\ka}\sum_{i=1}^m\left(\frac{\gamma^{(1)}_1}{u_i}+u_i\gamma^{(\infty)}_1\right)
\right)
\prod_{i=1}^m\left(
\frac{f^{(1)}}{u_i}+\frac{(f\otimes z)^{(1)}}{u_i^2}
\right){\bf 1_\gamma}, \nonumber
\end{align}
where $\Gamma$ is an appropriate cycle (see Proposition 4.2 in \cite{JNS}). 

We define linear isomorphisms $T_m: M_m\to \bigoplus_{i=0}^m\C x^i$ ($m\in\Z_{\ge 0}$) as 
\begin{equation*}
T_m\left(
(f^{(1)})^{i}\left((f\otimes z)^{(1)}\right)^{m-i}{\bf 1_\gamma}
\right)
=x^i\quad (0\le i\le m).
\end{equation*}
\begin{thm}\label{thm ckz iii}
For $m\in\Z_{\ge 0}$, we have 
\begin{equation*}
T_m\circ\left(
G_0^{(\infty)}-\frac{1}{4}h_0(h_0+2)+\frac{1}{2}\gamma^{(1)}_1\gamma^{(\infty)}_1
\right) =\frac{t}{\hbar^2}\widehat{H}_\mathrm{III}\left(
x,\hbar\frac{\partial}{\partial x}, m\hbar,b,t
\right)\circ T_m
\end{equation*}
as linear maps from $M_m$ to $ \bigoplus_{i=0}^m\C x^i$ 
with 
\begin{equation*}
\gamma^{(1)}_0=2(m-1)-\frac{b}{\hbar},\quad  \gamma^{(1)}_1=\frac{1}{\hbar},\quad  \gamma^{(\infty)}_1=-\frac{t}{\hbar}. 
\end{equation*} 
\end{thm}
\begin{proof}
It follows from direct computations. 
\end{proof}

As in the cases above, the Schr\"odinger equation for 
the quantum third Painlev\'e equation  has the following symmetry.  
\begin{prop}
If a function $\Phi^{\mathrm{III}}(x,t,a,b,\hbar)$ is a solution to the Schr\"odinger equation for 
the quantum third Painlev\'e equation, then 
the function $\Psi^{\mathrm{III}}(x,t,a,b,\hbar)$ defined as 
\begin{equation}\label{def h 3}
\Psi^{\mathrm{III}}(x,t,a,b,\hbar)=\Phi^{\mathrm{III}}(-x, t, -a,-b,\hbar)
\end{equation}
is a solution to 
\begin{equation*}
\hbar\frac{\partial}{\partial t}\Psi^\mathrm{III}(x,t)=\widehat{H'}_\mathrm{III} \Psi^\mathrm{III}(x,t), 
\end{equation*}
where $\widehat{H'}_\mathrm{III}$ are obtained by replacing $\hbar$ with $-\hbar$ in $\widehat{H}_\mathrm{III}$. 
\end{prop}
\begin{proof}
 It follows from direct computations. 
\end{proof}


\subsection{The case of $\mathrm{J=II}$}
Let $n=0$ and $r_\infty=3$. Then $M=M'(\gamma^{(\infty)}_1,\gamma^{(\infty)}_2,\gamma^{(\infty)}_3)$ and 
 the confluent KZ equation  for an unknown function $\Psi(\gamma^{(\infty)}_1)$ taking values in $M$ is 
\begin{equation}\label{eq 2}
\kappa \gamma^{(\infty)}_3\frac{\partial}{\partial \gamma^{(\infty)}_1} \Psi(\gamma^{(\infty)}_1)=
\left(
G^{(\infty)}_2-\frac{1}{4}\left(\gamma^{(\infty)}_1\right)^2-\frac{1}{2}\gamma^{(\infty)}_2+\frac{1}{2}h^{(\infty)}\gamma^{(\infty)}_2
\right)\Psi(\gamma^{(\infty)}_1), 
\end{equation}
where $\kappa\in\C$ and 
\begin{equation*}
G^{(\infty)}_2=\frac{1}{2}\left((e\otimes z)^{(\infty)}(f\otimes z)^{(\infty)}+(f\otimes z)^{(\infty)}(e\otimes z)^{(\infty)}
+\frac{1}{2}\left((h\otimes z)^{(\infty)}\right)^2\right). 
\end{equation*}

Let 
$M_m$ ($m\in\Z_{\ge 0}$) be the weight space of $M$ with the weight $-2m$, namely, 
\begin{align*}
M_m=&\left\{
v\in M\ \bigg|  \ h^{(\infty)}(v)=
-2mv
\right\}
\\=&\bigoplus_{i=0}^m\C ((f\otimes z)^{(\infty)})^i((f\otimes z^2)^{(\infty)})^{m-i}{\bf 1_\gamma}. 
\end{align*}
An integral formula taking values in $M_m$ for a solution to the confluent KZ equation \eqref{eq 2}  is 
\begin{align}
\int_\Gamma &\prod_{i=1}^mdu_i
\prod_{1\le i<j\le m}(u_i-u_j)^{2/\kappa}  \exp\left(
\frac{1}{\ka}\sum_{i=1}^m\left(u_i\gamma^{(\infty)}_1+\frac{u_i^2}{2}\gamma^{(\infty)}_2+\frac{u_i^3}{3}\gamma^{(\infty)}_3\right)
\right)\label{eq integral formula iii}
\\
&\times 
\prod_{i=1}^m\left(
(f\otimes z)^{(\infty)}+(f\otimes z^2)^{(\infty)}u_i
\right){\bf 1_\gamma}, \nonumber 
\end{align}
where $\Gamma$ is an appropriate cycle (see Proposition 4.2 in \cite{JNS}). 

We define linear isomorphisms $T_m: M_m\to \bigoplus_{i=0}^m\C x^i$ ($m\in\Z_{\ge 0}$) as 
\begin{equation*}
T_m\left(
((f\otimes z)^{(\infty)})^{i}((f\otimes z^2)^{(\infty)})^{m-i}{\bf 1_\gamma}
\right)
=x^i\quad (0\le i\le m).
\end{equation*}

\begin{thm}\label{thm ckz ii}
For $m\in\Z_{\ge 0}$, we have 
\begin{equation*}
T_m\circ\left(
G^{(\infty)}_2-\frac{1}{4}\left(\gamma^{(\infty)}_1\right)^2-\frac{1}{2}\gamma^{(\infty)}_2+\frac{1}{2}h^{(\infty)}\gamma^{(\infty)}_2
\right)=\frac{2}{\hbar^2}\widehat{H}_\mathrm{II}\circ T_m, 
\end{equation*}
as linear maps from $M_m$ to $\bigoplus_{i=0}^m\C x^i$  
with 
\begin{equation*}
 \gamma^{(\infty)}_1=\frac{t}{\hbar}, \quad \gamma^{(\infty)}_2=0, \quad \gamma^{(\infty)}_3=\frac{2}{\hbar}. 
\end{equation*}
\end{thm}
\begin{proof}
It follows from direct computations. 
\end{proof}


As in the cases above,  the Schr\"odinger equation for the quantum second Painlev\'e equation  has the following symmetry.  
\begin{prop}
If a function $\Phi^{\mathrm{II}}(x,t,a,\hbar)$ is a solution to  the Schr\"odinger equation for the quantum second Painlev\'e equation, then 
the function $\Psi^{\mathrm{II}}(x,t,a,\hbar)$ defined as 
\begin{equation}\label{def h 2}
\Psi^{\mathrm{II}}(x,t,a,\hbar)=\Phi^{\mathrm{II}}(-x, t, -a,\hbar)
\end{equation}
is a solution to 
\begin{equation*}
\hbar\frac{\partial}{\partial t}\Psi^\mathrm{II}(x,t)=\widehat{H'}_\mathrm{II} \Psi^\mathrm{II}(x,t), 
\end{equation*}
where $\widehat{H'}_\mathrm{II}$ are obtained by replacing $\hbar$ with $-\hbar$ in $\widehat{H}_\mathrm{IV}$. 
\end{prop}
\begin{proof}
 It follows from direct computations. 
\end{proof}


\bigskip
\bigskip
\noindent{\it \bf  Acknowledgement.}\quad

The author are grateful to 
 H. Kimura, Y. Ohta and Y. Yamada 
for helpful discussions and suggestions.

\bigskip
\bigskip



\begin{thebibliography}{[FJKLM]}


\bibitem{AFKMY}
 H.~Awata,  H.~Fuji, H.~Kanno, M.~Manabe and Y.~Yamada, 
 Localization with a Surface Operator, Irregular Conformal Blocks and Open Topological String, 
 	arXiv:1008.0574v3
\bibitem{BEMP}
M.~Berg\'ere, B.~Eynard, O.~Marchal, A.~Prats-Ferrer, 
Loop equations and topological recursion for the arbitrary-$\beta$ two-matrix model, 
arXiv:1106.0332v1

\bibitem{KZ EFK}
P.~I.~Etingof, I.~B.~Frenkel and A.~A.~Kirillov, Jr., 
Lectures on representation theory and Knizhnik-Zamolodchikov equations, 
{\em Math. Surveys
and Monographs} {\bf 58} {\em Amer. Math. Soc.}  (1998) 

\bibitem{FW}
P. J. Forrester and N. S. Witte, 
Application of the $\tau$-function theory of Painlev\'e equations to random matrices: $\mathrm{ P_{VI}}$, the JUE, CyUE, cJUE and scaled limits, 
{\em Nagoya Math. J.} 
{\bf 174} (2004), 29--114


\bibitem{Garnier}
R. Garnier, Sur des e\'quations diffe\'rentielles du troisie\`me ordre dont l'int\'egrale g\'en\'erale est uniforme et sur une classe d'\'equations nouvelles d'ordre 
sup\'erieur dont l'int\'egrale g\'en\'erale a ses points critiques fixes, {\em Ann. Sci. Ecole Norm. Sup.} {\bf 29} (1912), 1--126.


\bibitem{H}
J,~Harnad, 
Quantum isomonodromic deformations and the Knizhnik-Zamolodchikov
equations. Symmetries and integrability of difference equations 
(Est\'erel, PQ, 1994),
155--161, CRM Proc. Lecture Notes, 9, Amer.\ Math.\ Soc., Providence, RI, 1996

\bibitem{JNS}
M.~Jimbo, H.~Nagoya and J.~Sun,
Remarks on the confuent KZ equation for $\slt$ and
quantum Painlev\'e equations,
{\em J. Phys. A: Math. Theor.} {\bf 41} (2008)

\bibitem{JMU}
M.~Jimbo, T.~Miwa and K.~Ueno,
Monodromy preserving deformation of linear ordinary differential
equations with rational coefficients I: General theory and
$\tau$ function,
{\em Physica D} {\bf 2} (1981) 306--352


\bibitem{LO}
A.~M.~Levin and M.~A.~Olshanetsky, Painlev\'e-Calogero correspondence, 
arXiv:alg-geom/9706010v1

\bibitem{Masuda}
T.~Masuda, 
Classical transcendental solutions of the Painleve equations and their degeneration, 
{\em Tohoku Math. J.} {\bf 56} (2004) 467-490

\bibitem{N1}
H.~Nagoya,
Quantum Painlev\'e Systems of Type $A_l^{(1)}$,
 {\em Int.\ J.\ Math.}  {\bf 15} (2004), 1007--1031

%

\bibitem{N4}
H.~Nagoya,
A quantization of the sixth Painlev\'e equation, 
Noncommutativity and singularities, 291--298, {\em Adv. Stud. Pure Math.} {\bf 55} {\em Math. Soc. Japan, Tokyo} (2009)


\bibitem{NS} 
H.~Nagoya and J.~Sun, 
Confluent primary fields in the conformal field theory, 
  {\em J. Phys. A: Math. Theor.} {\bf 43} (2010) 



\bibitem{RNGT}
A.~Ramani, H.~Nagoya, B.~Grammaticos and T.~Tamizhmani, 
Folding transformations for quantum Painlev\'e equations, 
 {\em J.~Phys.~A: Math.~Theor.} {\bf 42} (2009) 

\bibitem{R}
N.~Reshetikhin,
The Knizhnik-Zamolodchikov system as a deformation of the isomonodromy problem,
{\em Lett. Math. Phys.} {\bf 26} (1992), 167--177

\bibitem{Suleimanov 1}
B. I. Suleimanov, 
The Hamiltonian structure of Painlev\'e equations and the method of isomonodromic deformations,
 {\em Differential Equations} {\bf 30} 726--732 (1994).


\bibitem{Szego}
    G.~Szeg\"o, Orthogonal polynomials, {\em AMS Colloquium publications} {\bf 23} (1939)

\bibitem{T}
K.~Takasaki, 
Painlev\'e-Calogero correspondence revisited,  
J. Math. Phys. 42 (2001), no. 3, 1443--1473

\end{thebibliography}
\end{document}